\documentclass[11pt]{article}
\usepackage{amsmath}
\usepackage{amsthm}
\usepackage{amssymb}
\usepackage{algorithm}
\usepackage{subfig}
\usepackage{color}
\usepackage[english]{babel}
\usepackage{graphicx}
\usepackage{wrapfig,epsfig}
\usepackage{epstopdf}
\usepackage{url}
\usepackage{graphicx}
\usepackage{color}
\usepackage{epstopdf}
\usepackage[noend]{algpseudocode}
\usepackage{scrextend}
\usepackage[T1]{fontenc}
\usepackage{bm}
\usepackage{bbm}
\usepackage{comment}
\usepackage{authblk}
\usepackage{multicol}
\usepackage{dsfont}

\usepackage{bbm}

\usepackage{tikz}
\usepackage{hyperref}  
\hypersetup{colorlinks=true,citecolor=blue,linkcolor=blue}
 \usetikzlibrary{arrows}
\usepackage[margin=1in]{geometry}
\graphicspath{{./figs/}}

\author{
	Karl Bringmann\thanks{\texttt{kbringma@mpi-inf.mpg.de} Max Planck Institute for Informatics, Saarland Informatics Campus, Saarbr\"ucken, Germany.}
	\quad
	Vasileios Nakos\thanks{\texttt{vnakos@mpi-inf.mpg.de}, Max Planck Institute for Informatics, Saarland Informatics Campus, Saarbr\"ucken, Germany.} 
}

\date{}

\title{Fast $n$-fold Boolean Convolution via Additive Combinatorics\footnote{This work is part of the project TIPEA that has received funding from the European Research Council (ERC)
under the European Unions Horizon 2020 research and innovation programme (grant agreement No. 850979).}}

\newtheorem{theorem}{Theorem}[section]
\newtheorem{lemma}[theorem]{Lemma}
\newtheorem{definition}[theorem]{Definition}

\newtheorem{corollary}[theorem]{Corollary}

\newtheorem{remark}[theorem]{Remark}
\newtheorem{claim}[theorem]{Claim}

\makeatletter
\newcommand{\ostar}{\mathbin{\mathpalette\make@circled\star}}
\newcommand{\make@circled}[2]{%
  \ooalign{$\m@th#1\smallbigcirc{#1}$\cr\hidewidth$\m@th#1#2$\hidewidth\cr}%
}
\newcommand{\smallbigcirc}[1]{%
  \vcenter{\hbox{\scalebox{0.77778}{$\m@th#1\bigcirc$}}}%
}
\makeatother

\newcommand{\wt}{\widetilde}
\newcommand{\tOh}{\widetilde{O}}

\newcommand{\conv}{\star}
\newcommand{\convBool}{\ostar}

\newcommand{\R}{\mathbb{R}}

\renewcommand{\varepsilon}{\epsilon}
\renewcommand{\tilde}{\wt}

\newcommand{\Z}{{\mathbb{Z}}}

\DeclareMathOperator*{\polylog}{polylog}

\DeclareMathOperator{\poly}{poly}

\makeatletter
\newcommand*{\RN}[1]{\expandafter\@slowromancap\romannumeral #1@}
\makeatother

\usepackage{lineno}

\usepackage{etoolbox}

\makeatletter

\newcommand{\define}[4][ignore]{%
  \ifstrequal{#1}{ignore}{}{
  \@namedef{thmtitle@#2}{#1}}%
  \@namedef{thm@#2}{#4}%
  \@namedef{thmtypen@#2}{lemma}%
  \newtheorem{thmtype@#2}[theorem]{#3}%
  \newtheorem*{thmtypealt@#2}{#3~\ref{#2}}%
}

\newcommand{\state}[1]{%
  \@namedef{curthm}{#1}
  \@ifundefined{thmtitle@#1}{
  \begin{thmtype@#1}
    }{
  \begin{thmtype@#1}[\@nameuse{thmtitle@#1}]
  }
    \label{#1}
    \@nameuse{thm@#1}
  \end{thmtype@#1}
  \@ifundefined{thmdone@#1}{
  \@namedef{thmdone@#1}{stated}%
  }{}
}

\newcommand{\restate}[1]{%
  \@namedef{curthm}{#1}
  \@ifundefined{thmtitle@#1}{
    \begin{thmtypealt@#1}
    }{
  \begin{thmtypealt@#1}[\@nameuse{thmtitle@#1}]
  }
    \@nameuse{thm@#1}
  \end{thmtypealt@#1}
  \@ifundefined{thmdone@#1}{
  \@namedef{thmdone@#1}{stated}%
  }{}
}

\newcommand{\thmlabel}[1]{
  \@ifundefined{thmdone@\@nameuse{curthm}}{\label{#1}
    }{\tag*{\eqref{#1}}}
}
\makeatother

\begin{document}

\begin{titlepage}
  \maketitle
  \begin{abstract}

We consider the problem of computing the Boolean convolution (with wraparound) of $n$~vectors of dimension $m$, or, equivalently, the problem of computing the sumset $A_1+A_2+\ldots+A_n$ for $A_1,\ldots,A_n \subseteq \mathbb{Z}_m$. 
Boolean convolution formalizes the frequent task of combining two subproblems, where the whole problem has a solution of size $k$ if for some $i$ the first subproblem has a solution of size~$i$ and the second subproblem has a solution of size $k-i$. Our problem formalizes a natural generalization, namely combining solutions of $n$ subproblems subject to a modular constraint. This simultaneously generalises Modular Subset Sum and Boolean Convolution (Sumset Computation). Although nearly optimal algorithms are known for special cases of this problem, not even tiny improvements are known for the general case. 

We almost resolve the computational complexity of this problem, shaving essentially a factor of $n$ from the running time of previous algorithms.
Specifically, we present a \emph{deterministic} algorithm running in \emph{almost} linear time with respect to the input plus output size $k$. 
We also present a \emph{Las Vegas} algorithm running in \emph{nearly} linear expected time with respect to the input plus output size $k$. 
Previously, no deterministic or randomized $o(nk)$ algorithm was known. 

At the heart of our approach lies a careful usage of Kneser's theorem from Additive Combinatorics, and a new deterministic almost linear output-sensitive algorithm for non-negative sparse convolution. In total, our work builds a solid toolbox that could be of independent interest.

  \end{abstract}
  \thispagestyle{empty}
\end{titlepage}

\newpage


\section{Introduction}

In this paper we study $n$-fold variants of the following fundamental 2-fold problems.

\subsection{2-Fold Case}

\subparagraph*{Boolean Convolution and Sumset Computation.}
In Boolean convolution we are given vectors $A,B \in \{0,1\}^m$ and the task is to compute the vector $C = A \convBool B \in \{0,1\}^m$ defined by $C[k] = \bigvee_i A[i] \wedge B[k-i]$. This formalizes a situation in which we split a computational problem into two subproblems, so that in total there is a solution of size $k$ if and only if for some $i$ there is a solution of the left subproblem of size $i$ and there is a solution of the right subproblem of size $k-i$. 
This is a natural task that frequently arises in algorithm design.
There are two variants of this problem: \emph{Without wraparound} the $\bigvee_i$-quantifier goes over $0 \le i \le k$; \emph{with wraparound} the quantifier goes over all $i \in [m]$ and the entry $B[k-i]$ means $B[(k-i) \bmod m]$. Algorithmically the two variants are equivalent, and throughout this paper we study the latter variant.

An equivalent problem is \emph{sumset computation}: Given sets $A,B \subseteq \Z_m$, compute their sumset $A+B$, which denotes the set of all sums $a+b$ modulo $m$ with $a \in A, b \in B$. This corresponds to Boolean convolution with wraparound\footnote{\label{foot:one} By removing the modulo operation and thus working over $\Z$ we can also pose a problem variant corresponding to Boolean convolution without wraparound. Again, algorithmically these variants are equivalent, since for any $A,B \subseteq \{0,1,\ldots,m-1\}$, on the one hand computing $A+B$ over $\Z$ and taking the result modulo $m$ yields $A+B$ over $\Z_m$, and on the other hand computing $A+B$ over $\Z_{2m}$ yields $A+B$ over $\Z$.}. 

\subparagraph*{Standard Convolution and Polynomial Multiplication.}
In (standard) convolution we are given vectors $A,B \in \R^m$ and the task is to compute the vector $C = A \conv B \in \R^m$ with $C[k] = \sum_i A[i] \cdot B[k-i]$. For instance, if $A[i]$ and $B[i]$ count the number of size-$i$ solutions of the left and right subproblem, then $C[k]$ counts the number of size-$k$ solutions of the whole problem.
Again one can consider variants with or without wraparound.
A typical restriction are \emph{non-negative} entries, which is well-motivated in case that $A,B,C$ represent numbers of solutions.

This problem is equivalent to \emph{polynomial multiplication}: Given the coefficients of polynomials $P(X) = \sum_{i=0}^m A[i] \cdot X^i$ and $Q(X) = \sum_{i=0}^m B[i] \cdot X^i$, compute the coefficients of their product $P \cdot Q$.

\subparagraph*{State of the Art.}
Using Fast Fourier Transform (FFT), all of the above problems can be solved in time $O(m \log m)$. 
A long line of work has considered these problems in a sparse setting, called sparse convolution or sparse polynomial multiplication, see, e.g.,~\cite{mut95,CH02,roche2008adaptive,monagan2009parallel,van2012complexity,AR15,CL15,roche2018can,N19,giorgi20,BFN21}. Here the task is to compute the convolution of two sparse vectors much faster than performing FFT, ideally in near-linear time in terms of the input plus output size (i.e., the number of non-zero entries of the input and output vectors). Near-linear in the input plus output size running time was achieved for vectors with non-negative entries by Cole and Hariharan~\cite{CH02} and for general vectors in~\cite{N19}, see also~\cite{giorgi20} for additional $\log m$ factors improvements. Very recently, a Monte Carlo $O(k \log k)$-time algorithm has been achieved in~\cite{BFN21} for non-negative convolution, where $k$ is the input plus output size. 
Sparse convolution techniques are crucially used in~\cite{AKP07,AKPR14,abp14,CL15,ABJTW19,BN20}, and are also relevant to the study of sparse wildcard matching, a fundamental string problem~\cite{cs98,CH02}.

However, all known algorithms for these sparse problems are randomized, and thus an open problem is to \emph{close the gap between deterministic and randomized algorithms}. This was explicitly posed as an open problem in~\cite[Remark 8.2]{CL15}.

\subparagraph{Our Contribution to the 2-Fold Case.}
We present a deterministic algorithm for convolution of non-negative vectors (and thus also for Boolean convolution) running in time $k \cdot m^{o(1)}$, where $k$ is the input plus output size. 
This matches up to the $m^{o(1)}$ term the best known algorithms in the randomize case~\cite{CH02}. Our algorithm heavily builds upon an algorithm by Chan and Lewenstein~\cite{CL15}, which operates under the additional assumption that a small superset of the non-negative terms is known in advance. We remove their assumption by gradually building the sumset using calls to their algorithm. 
\begin{theorem}[Deterministic Non-Negative Sparse Convolution, Section~\ref{sec:sumset_comp}] \label{thm:sumset_deterministic}
Denote by $\|x\|_0$ the number of non-zero entries of a vector $x$.
Given vectors $A,B \in \mathbb{R}_{\ge 0}^m$, we can compute their convolution $A \conv B$ (with wraparound) in time $\| A \conv B \|_0 \cdot m^{o(1)}$ by a deterministic algorithm. More precisely, the running time is
		$ \|A \conv B\|_0\cdot 2^{ O(\sqrt{ \log \|A \conv B\|_0 \log \log m}) }$.
\end{theorem}
Observe that $\|A\|_0, \|B\|_0 \le \|A \conv B\|_0$, and thus rather than bounding the running time in terms of the input plus output size $\|A\|_0 + \|B\|_0 + \|A \conv B\|_0$, it suffices to bound the running time in terms of only the output size $\|A \conv B\|_0$.
Moreover, note that since $\|A \conv B\|_0 \le m$ the above running time is bounded by $m^{1+o(1)}$. As an additional bonus, our approach gives a quite simple $\|A \conv B\|_0 \cdot \polylog(m)$-time Las Vegas algorithm for the $2$-fold case of non-negative sparse convolution, see Theorem~\ref{lemma:nonneg_conv}.

We leave it as an open problem whether similarly efficient deterministic algorithms exist under the presence of negative entries.

\subsection{\boldmath$n$-Fold Case}

The focus of this paper is on $n$-fold generalizations of the above problems. 
Indeed, in typical applications we do not only split a problem into two subproblems, but these subproblems are recursively split into further subproblems. 
If the recursion tree has $n$ leaves, we therefore want to compute Boolean convolutions of the form $A_1 \convBool \ldots \convBool A_n$ for vectors $A_1,\ldots,A_n$.

Note that now the ``gold standard'' would be linear running time in terms of the total input plus output size $k = \|A_1\|_0+\ldots+\|A_n\|_0+\|A_1 \convBool \ldots \convBool A_n\|_0$. Note that in contrast to the $2$-fold case, the size of the output is incomparable to the size of the input. 

\subparagraph{A Special Case: Modular Subset Sum}
As an example, consider the Modular Subset Sum problem, where we are given $x_1,\ldots,x_n \in \Z_m$ and a target $t$, and the task is to decide whether for some subset $I \subseteq [n]$ we have $\sum_{i \in I} x_i \equiv t \pmod m$. Observe that the sumset $\{0,x_1\}+\ldots+\{0,x_n\} \subseteq \Z_m$ denotes the set of all attainable subset sums modulo $m$, and thus Modular Subset Sum can be solved by a direct application of $n$-fold sumset computation, which is equivalent to $n$-fold Boolean convolution (with wraparound). 

The state of the art for Modular Subset Sum is as follows. A standard dynamic programming approach solves the problem in time $O(n \cdot m)$. After the first improvements by Koiliaris and Xu~\cite{KoiliarisX17}, Axiotis et al.~\cite{ABJTW19} designed an algorithm running in time $O((n+m) \polylog (n+m))$, which was further simplified, sped up and made deterministic in~\cite{axiotis2021fast,cardinal2021modular}. Those running times match a conditional lower bound based on the Strong Exponential Time Hypothesis~\cite{ABHDD19,ABJTW19}. Moreover, all of the above algorithms can be analyzed to run in time $O(k \polylog m)$, where $k$ is the total input plus output size~\cite{ABJTW19}. 

In other words, for the special case $|A_1| = \ldots = |A_n| = 2$ of $n$-fold sumset computation near-optimal algorithms are known. Furthemore, the techniques crucially exploit the fact that all sets $A_i$ have constant cardinality. The goal of this paper is to investigate the general case without any restrictions on $|A_i|$. Can one move beyond the problem-specific techniques in~\cite{KoiliarisX17,ABJTW19,axiotis2021fast,cardinal2021modular} which seem to apply solely to Modular Subset Sum?

\subparagraph{Naive Approach}
As it is already known how to compute $A \conv B$ (and thus $A \convBool B$) in time near-linear in the output size~\cite{CH02,N19,BFN21}, is there an easy generalization to compute the $n$-fold Boolean convolution $A_1 \convBool \ldots \convBool A_n$? 
Naively, if we compute the $n$-fold convolution in a linear fashion as $(((A_1 \convBool A_2) \convBool A_3) \conv \ldots \convBool A_{n-1}) \convBool A_n$, then each intermediate convolution has input plus output size at most $k$, so using~\cite{CH02} we can bound the total expected running time by $O(nk \polylog m)$. 
Unfortunately, this running time analysis is tight. The issue is that up to $\wt \Omega(n)$ intermediate results may have size $\Omega(k)$.

The same is true if we compute the $n$-fold convolution in a bottom-up tree-like fashion as $((A_1 \convBool A_2) \convBool (A_3 \convBool A_4)) \convBool \ldots$, as shown by the following example. Pick $b = \big\lfloor \frac{ \log m }{ \log \log m} \big \rfloor$ and $\ell = \lceil \log_b(m) \rceil = \Theta(b)$, and set $A_i$ to the indicator vector of $ b^{ i \bmod \ell} \cdot \{0,1,2,\ldots,b-1\}$. Then the Boolean convolution of any $\ell$ consecutive $A_i$'s is the all-ones vector and thus has size $m$, and this holds for $\Omega(n/ \ell) = \Omega\big(n \frac{\log \log m}{\log m}\big)$ intermediate convolutions. On the other hand, the input size is $O(n \frac{\log m}{\log \log m})$ and the output size is $O(m)$.


Analyzing the time only in terms of $n,m$, the naive approach yields time $O(nm \polylog m)$. A simple algorithm using $n-1$ FFTs also yields time $O(nm \log m)$. Before this work it was open whether $n$-fold Boolean convolution can be solved in time close to linear in $n+m$, and close to linear in $k$, or whether the additional factor $\wt \Theta(n)$ of the naive approach is necessary.


\subparagraph{\boldmath$n$-Fold Boolean Convolution versus \boldmath$n$-Fold Convolution.} We note that $n$-fold Boolean convolution is quite different from $n$-fold convolution, and we focus on the former in this paper. The reason is that $n$-fold convolution results in exponentially large entries. Indeed, assuming that $A_1,\ldots,A_n$ are non-negative integer vectors, each with at least two non-zero entries, one can check that $\|A_1 \conv \ldots \conv A_n\|_1 = \|A_1\|_1 \cdot \ldots \cdot \|A_n\|_1 \ge 2^n$ (here $\|x\|_1 = \sum_i|x[i]|$), and thus at least one output entry requires $\Omega(n)$ bits to represent exactly. 
Possible ways to handle this situation are (1) to let $k$ be the total number of input plus output \emph{bits}, (2) assume that entries come from a finite field, or (3) relax to approximation. 
We leave these as open problems and focus on Boolean convolution in this paper.


\subparagraph{Our Contribution to \boldmath$n$-Fold Boolean Convolution.}
We show that the multiplicative factor $n$ in the naive running times $O( n k \polylog m)$ and $O( s+ nm \log m)$ is not necessary (here $s$ is the size of the input). 
Specifically, our approach yields two new results for $n$-fold Boolean convolution: a randomized Las Vegas algorithm running in expected time $O(k \cdot \polylog m)$, and a deterministic algorithm running in time $k \cdot m^{o(1)}$. Morally, we show that one can convolve $n$ Boolean vectors in a much better way than doing $n-1$ FFTs. In particular, in terms of $m,n$ and the size of the input $s$, the known algorithms would run in time $\tOh(s + mn)$, whereas our approach yields time $\tOh(s+m)$. Thus, in instances where the size of the input does not dominate (as in Modular Subset Sum where $s=2n$) our approach yields a substantial improvement.


Our algorithm falls in a line of research that tries to apply results from Additive Combinatorics in algorithm design, such as \cite{GalilM91,CL15,AKKN16, BGNV17,mwm19,BN20}. Quite interestingly, this is the first time that such a connection has produced an (almost) optimal result. Previous algorithms~\cite{GalilM91,CL15,AKKN16, BGNV17, mwm19, BN20} had less clean running time bounds and are thus likely to be suboptimal, partly because of the Additive Combinatorics machinery used. 

\smallskip
We now state our results more formally. 
The main result of this paper is the following. 

\begin{theorem}[$n$-Fold Boolean Convolution] \label{thm:n_fold_conv}
Given vectors $A_1,A_2,\ldots,A_n \in \{0,1\}^m$ we can compute their Boolean convolution with wrap-around $A_1 \convBool A_2 \convBool \ldots \convBool A_n$
\begin{enumerate}
\item[(1)] by a randomized Las Vegas algorithm in $O(k \cdot \polylog (mk))$ expected time, or
\item[(2)] by a deterministic algorithm in $k \cdot  2^{O(\sqrt{\log k \cdot \log \log m})}  $ time
\end{enumerate}
Here, $k := \|A_1\|_0 + \ldots + \|A_n\|_0 + \| A_1 \convBool A_2 \convBool \ldots \convBool A_n \|_0$ is the total input plus output size.
\end{theorem}

\begin{remark}
It might seem confusing that for very small $k$, specifically for $k \le \log^{O(1)} m$, our deterministic time is faster than our randomized time. However, as we will discuss later, it is easy to solve the problem deterministically in time $k^{O(1)}$. In fact our time bounds are $\mathrm{min}\left\{k^3,k \cdot \polylog (mk))\right\}$ expected time, and $\mathrm{min}\left\{k^3,k \cdot  2^{O(\sqrt{\log k \cdot \log \log m})} \cdot \polylog(mk)\right\}$ deterministically; the latter can be simplified to the expression in Theorem~\ref{thm:n_fold_conv}.
\end{remark}

In order to employ Additive Combinatorics machinery, it will be convenient to phrase the problem in terms of sets and sumsets, making the connection more clear.
To this end, we replace every vector $A_i \in \{0,1\}^m$ by a set $A'_i \subseteq \mathbb{Z}_m$ such that $x \in A'_i$ if and only the $x$-th entry of $A_i$ is $1$. Then, it can easily be seen that the Boolean convolution $A_1 \convBool A_2 \convBool \ldots \convBool A_n$ is equivalent to computing the sumset $A'_1+A'_2+\ldots + A'_n$. Written in a more Additive-Combinatorics-friendly way, our main result can be rephrased in the following way.

\begin{theorem}[Theorem \ref{thm:n_fold_conv} restated, Section~\ref{sec:n_fold_general}]\label{thm:n_fold_comp}
Given sets $A_1,\ldots A_n \subseteq \mathbb{Z}_m$, we can compute their sumset $A_1+A_2 + \ldots + A_n$
\begin{enumerate}
\item[(1)] by a randomized Las Vegas algorithm in $O(k \cdot \polylog (mk))$ expected time, or
\item[(2)] by a deterministic algorithm in $k \cdot  2^{O(\sqrt{\log k \cdot \log \log m})} $ time.
\end{enumerate}
 Here, $k := |A_1|+\ldots+|A_n| + |A_1+\ldots+A_n|$ is the total input plus output size.
\end{theorem}

We remark that further improvements over Theorem~\ref{thm:sumset_deterministic} would directly improve Theorems~\ref{thm:n_fold_conv} and~\ref{thm:n_fold_comp}. In particular, our factor $m^{o(1)} = 2^{O(\sqrt{\log m \log \log m})}$ stems entirely from the application of Theorem~\ref{thm:sumset_deterministic}, and thus indirectly from a tool called the FFT Lemma~\cite{CL15} that we use to prove Theorem~\ref{thm:sumset_deterministic}.

We also remark that Theorem~\ref{thm:n_fold_comp} is formulated for sumsets over $\Z_m$, but by setting $m$ sufficiently large (like $1+\sum_i \max(A_i)$) we can also compute sumsets $A_1+\ldots+A_n \subseteq \Z$ over the integers in time close to the input plus output size. However, this is a much simpler result that can also be achieved by elementary means, without any Additive Combinatorics.

\section{Preliminaries and Technical Toolkit}

For any positive integer $m$, we let $\Z_m$ be the group of residues modulo~$m$.
For two sets $A,B \subseteq \Z_m$, we define $A+B := \{ x \mid \exists a\in A, b \in B \colon a+b = x\}$. 
Unless explicitly stated otherwise, all sumsets throughout the paper are computed in the underlying group~$\Z_m$, i.e., $A+B \subseteq \Z_m$.
We also write $A \bmod q := \{ a \bmod q \mid a \in A\}$. 

Throughout the paper we use the notation of sumset computation instead of the equivalent Boolean convolution.




\subsection{Randomized Sumset Computation}

Cole and Hariharan's sparse convolution algorithm~\cite{CH02} implies that the sumset $A+B$ can be computed in Las Vegas time $O(|A+B| \cdot \log^2 m + \poly(\log m))$. Very recently, this was improved to $O(|A+B| \cdot \log |A+B| + \poly(\log m))$~\cite{BFN21} with a Monte Carlo algorithm.

\begin{theorem}[Randomized Sumset Computation, \cite{CH02}, see also Section~\ref{sec:sumset_comp}]\label{thm:sumset_randomized}
Given sets $A,B \subseteq \Z_m$, their sumset $A+B$ can be computed in expected time $O( |A+B| \poly(\log m) )$.
\end{theorem}

\subsection{The Symmetry Group and its Properties} \label{sec:toolkit_sym}

\begin{definition}[Symmetry group of a set]
Let $A \subseteq \Z_m$. We define the symmetry group of $A$ as $Sym(A) = \{ h \in \Z_m \mid A+\{h\} = A\}$. 
\end{definition}

It is easy to check that $Sym(A)$ satisfies the group properties with respect to addition, and thus $Sym(A)$ is a subgroup of $\mathbb{Z}_m$. In particular, we have $Sym(A) = d \cdot \mathbb{Z}_{m/d}$, where $d$ is the minimum non-zero element of $Sym(A)$ (to see this, note that the minimum non-zero element of a cyclic subgroup is also a generator of it).

One can check that $Sym(A) \subseteq Sym(A+B)$ holds for any sets $A,B \subseteq \Z_m$. This property will be of great importance to us. Moreover, for any non-empty set $A$ and any $x \in A$ we have $Sym(A) \subseteq A + \{-x\}$. This holds since any $h \in Sym(A)$ maps $x$ to some $x' \in A$, which means $x' = x + h \pmod m$, hence $h = x' - x \pmod m$. In particular, the symmetry group of a non-empty set $A$ has size at most $|A|$. 

We show that the symmetry group can be computed in linear time.

\begin{theorem}[Computing the Symmetry Group, Section~\ref{sec:symmetry_group}]\label{thm:symmetry_group}
Given a sorted non-empty set $A \subseteq \Z_m$, we can compute $Sym(A)$ in time $O(|A|)$.
\end{theorem}

\subsection{Kneser's Theorem}

The following theorem lies at the core of our algorithms.

\begin{theorem}[Kneser's Theorem, see, e.g., Theorem 5.5 in~\cite{tao2006additive}] \label{thm:kneser} Let $A,B \subseteq \mathbb{Z}_m$ be non-empty. Then
\[	|A+B| \geq \min \{ |A|+|B| - |Sym(A+B)|, m \}.	\]
\end{theorem}

We will use the following simple corollary.
\begin{corollary}\label{cor:kneser}
  Let $A,B \subseteq \Z_m$ be non-empty. If $|A+B| < |A|+|B|-1$ then $|Sym(A+B)|>1$.
\end{corollary}
\begin{proof}
  For $m=1$ it cannot happen that $|A+B| < |A|+|B|-1$, so assume $m \ge 2$. 
  
  If $|A+B| = m$, then $A+B = \Z_m$. This implies $Sym(A+B)=\Z_m$ and thus $|Sym(A+B)| = m > 1$.
  Otherwise, if $|A+B| < m$, we can simplify the bound obtained from Kneser's Theorem to
  \[ |A+B| \ge |A|+|B| - |Sym(A+B)|. \]
  Together with $|A+B| < |A|+|B|-1$, this implies $|Sym(A+B)| > 1$.
\end{proof}

%

\section{Overview and Comparison with Previous Approaches}

We start by giving a rough overview of our algorithm, leaving out several details.
Our improvements are obtained by delving deeper into the additive structure of sumset computation over $\mathbb{Z}_m$ than previous work. Our algorithms compute the sumset $A_1+\ldots+A_n$ in a bottom-up 
%
tree-like fashion as $((A_1+A_2)+(A_3+A_4)) + \ldots$. For any two sets $X,Y$ for which we compute $X+Y$ during the execution of this algorithm, we check whether $|X+Y| < |X|+|Y|-1$. If this is the case, Kneser's Theorem (specifically Corollary~\ref{cor:kneser}) implies that $X+Y$ has a non-trivial symmetry group, and hence $A_1+\ldots+A_n$ has a non-trivial symmetry group. A non-trivial symmetry group of a set $Z = X+Y \subseteq \mathbb{Z}_m$ implies that the set is \emph{periodic}: there exists a divisor $d$ of $m$ and a set $Z' \subseteq \{0,\ldots, d-1\}$ such that $Z = Z' + d \cdot \mathbb{Z}_{m/d} = Z'+\{0,d,2d,\ldots,m-d\}$, i.e., $Z$ is a rotation (by multiples of $d$) of a subset of $\{0,\ldots,d-1\}$. 
This allows us to reduce to the smaller universe $\Z_d$, which is progress (it might seem from this discussion that we require a factorization of $m$, but this is not the case: if we reduce to a smaller universe $\Z_d$, then $d$ is a divisor of $m$ that can be easily read off the sumset $Z = X+Y$, by computing the symmetry group $Sym(X+Y)$ and taking its smallest non-zero element). 
It remains to argue about the situation in which every computed sumset satisfies $|X+Y| \ge |X|+|Y|-1$. Using this inequality, we can control at any intermediate step of the algorithm the total size of all sumsets computed so far. When the computation arrives at the root, the running time that we spent on computing these sumsets is almost linear in the input plus output size.


\subparagraph{Why Previous Approaches Cannot Solve the Generalized Problem.}

A natural question to ask is whether previous algorithms for Subset Sum or Modular Subset Sum were also able to tackle the more general problem of $n$-fold sumset computation. The techniques underlying the algorithms for Subset Sum in \cite{Bring17,KoiliarisX17, JH19} are inherently non-modular, and hence cannot facilitate $n$-fold sumset computation problem over the group $\mathbb{Z}_m$. More relevant is the Modular Subset Sum problem, which is a standard variant of Subset Sum, where one works over $\mathbb{Z}_m$ rather than $\mathbb{Z}$. This problem has seen two interesting developments in the last few years.

The deterministic algorithm of Koiliaris and Xu~\cite{KoiliarisX17} uses multiple interesting problem-specific tricks for Modular Subset Sum, but it is unclear how to generalize them to $n$-fold sumset computation. In fact, their algorithm can be viewed as a reduction from Modular Subset Sum to $\mathrm{min}\{\sqrt{n}, m^{1/4}\}$-fold sumset computation, which they then solve by the straightforward repeated Fast Fourier Transform. Hence, also for the general case of $n$-fold sumset computation their approach does not seem to yield time $o(nm)$. 


All known algorithms for Modular Subset Sum~\cite{ABJTW19,axiotis2021fast,cardinal2021modular} compute the set of attainable subset sums $\mathcal{S}(A) = \{0,a_1\}+\ldots+\{0,a_n\} \subseteq \Z_m$ for $A = \{a_1,\ldots,a_n\}$.
The main idea is to compute $\mathcal{S}(A)$ from $\mathcal{S}(A \setminus \{a\})$ by forming the vector $1_{a + \mathcal{S}(A\setminus \{a\})} - 1_{\mathcal{S} (A\setminus \{a\}) }$. 
It can be easily seen that this vector consists of an equal number of positive and negative entries, and the positive entries correspond to the ``new'' sums $\mathcal{S}(A) \setminus \mathcal{S}(A \setminus \{a\})$.
Using hashing-based arguments or appropriate data structures for string manipulation, they show how to recover the support of the aforementioned vector in near-linear output-sensitive time. 
A possibility to generalize this approach to $n$-fold sumset computation $A_1+\ldots+A_n$ would be to consider the vector $\sum_{a \in A_n} (1_{a + A_1+\ldots+A_{n-1}} - 1_{A_1+\ldots+A_{n-1}})$. However, measuring this vector would incur time $\Omega(|A_n|)$, and thus an immediate generalization of their approach would at least pay a factor $\max_i |A_i|$ on top of the output size.

\subparagraph{Symmetry Manifestations in Higher Dimensions.} It would be interesting to understand whether the symmetry considerations of our algorithm manifest themselves in other abelian groups, most notably in $\mathbb{Z}_m^D$. The characterization of subgroups over $\mathbb{Z}_m^D$ with $D>1$ is less convenient for our purposes than the characterization in the one-dimensional case, so it seems that a different treatment and notion of progress is needed in that case. We leave this to potential future work. Even if one does worry about the $n$-fold case and concentrates in the simplest case of $n=2$, i.e. $2$-fold $d$-dimensional sparse convolution, the best algorithm we are aware of solves the problem with a multiplicative $2^d$ multiplicative factor on top of output size. We leave as an open question the problem of avoiding the exponential dependence of $2$-fold $d$-dimensional sparse convolution.

\section{Warmup: \boldmath$n$-Fold Sumset Computation over Prime Universe}
\label{app:warmup}

As a warmup, we consider universe $\Z_m = \Z_p$ for prime $p$. For simplicity, we analyze our algorithm only in terms of the input size and the universe size $p$, that is, we defer the output-sensitive analysis to the general algorithm in Section~\ref{sec:n_fold_general}.

Suppose we are given sets $A_1,\ldots,A_n \subseteq \Z_p$. We may assume that $n$ is a power of $2$, since otherwise we can add an appropriate number of sets $A_i = \{0\}$ without affecting the sumset. 
Consider Algorithm~\ref{alg:prime_sum}. We compute $A_1+\ldots+A_n$ in a tree-like bottom-up fashion, by first computing $A_1+A_2, A_3+A_4,\ldots$, then computing $A_1+A_2+A_3+A_4, \ldots$, and so on. The intermediate sets in round $r$ are called $X_{r,1},\ldots,X_{r,n/2^r}$. The termination criterion is that the sets that we computed so far in the current round $r$ have total size significantly more than $p$, more precisely, $\sum_{j\le i} |X_{r,j}| > p + i - 1$. If this criterion is satisfied, then we return the complete universe $\Z_p$.
If the termination criterion is never satisfied, then in the end we return $X_{\log n, 1}$.

\begin{algorithm}[!t]\caption{}\label{alg:prime_sum}
\begin{algorithmic}[1]
\Procedure{\textsc{nFoldSumsetInPrimeUniverse}}{$A_1,A_2,\ldots,A_n,p$} 
\\ \Comment{$n$ is a power of $2$; $p$ is prime; non-empty sets $A_1,\ldots,A_n \subseteq \Z_p$}
\State $X_{0,i} \leftarrow A_i$, for all $i \in [n]$
\For { $r=1$ to $ \log n$}
	\For { $i=1$ to $ n/2^r $}
		\State $X_{r,i} \leftarrow X_{r-1,2i-1} + X_{r-1,2i}$ \Comment{sumset computation via Theorem \ref{thm:sumset_deterministic}}
	\If { $\sum_{j \leq i } |X_{r,j} | > p + i -1$ }
		\State \Return $\Z_p$
	\EndIf
	\EndFor
\EndFor
\State \Return $X_{\log n, 1}$
\EndProcedure
\end{algorithmic}
\end{algorithm}

It remains to analyze correctness and running time of this algorithm. To analyze correctness of the termination criterion, we need the following lemma.

\begin{lemma}\label{lem:cauchy}
Let $p$ be a prime, and let $A_1,A_2,\ldots, A_n \subseteq \mathbb{Z}_p$ be non-empty. If $ \sum_{j=1}^n |A_j| \geq p + n-1$, then $A_1 + A_2+ \ldots + A_n =\mathbb{Z}_p$.
\end{lemma}

\begin{proof}
Suppose that the symmetry group has size $|Sym(A_1+\ldots+A_n)| > 1$. Since $\mathbb{Z}_p$ has no non-trivial subgroups, this yields $Sym(A_1+\ldots+A_n) = \mathbb{Z}_p$. Since $|Sym(A)| \le |A|$ holds for any set $A$, we obtain $A_1+\ldots+A_n = \mathbb{Z}_p$.

It remains to consider the case $|Sym(A_1+\ldots+A_n)| = 1$. Since $Sym(A) \subseteq Sym(A+B)$ holds for any sets $A,B$, it follows that $|Sym(A_1+\ldots+A_i)| = 1$ for all $1 \le i \le n$. 
We now inductively prove that $|A_1+\ldots+A_i | \geq \min\{\sum_{j=1}^i |A_j| - i +1,p\}$, from which the corollary follows.
The induction base for $i=1$ is trivial. For $i > 1$, we use Kneser's theorem on $A := A_1+\ldots+A_{i-1}$ and $B := A_i$ to obtain
\[ |A_1+\ldots+A_i| \ge \min\left\{ |A_1+\ldots+A_{i-1}| + |A_i| - |Sym(A_1+\ldots+A_i)|, p \right\}. \]
Plugging in $|Sym(A_1+\ldots+A_i)| = 1$ and the induction hypothesis on $|A_1+\ldots+A_{i-1}|$, and simplifying $\min\{\min\{a,p\}+b,p\}$ to $\min\{a+b,p\}$, yields
\[ |A_1+\ldots+A_i| \ge \min\Big\{ \Big(\sum_{j=1}^{i-1} |A_j| - (i-1) +1\Big) + |A_i| - 1, p \Big\} = \min\Big\{ \sum_{j=1}^i |A_j| - i +1, p \Big\}, \]
which finishes the inductive proof.\footnote{We remark that for this lemma it would be sufficient to use the Cauchy-Davenport theorem (see, e.g., \cite[Theorem 5.4]{tao2006additive}) instead of Kneser's theorem. Only for the generalization to non-prime $m$ we need the more general theorem by Kneser.}
%

\end{proof}

\begin{lemma}[Analysis of Algorithm~\ref{alg:prime_sum}]
Given non-empty sets $A_1,\ldots,A_n \subseteq \Z_p$, where $p$ is prime and $n$ is a power of 2,
Algorithm~\ref{alg:prime_sum} correctly computes $A_1+\ldots+A_n$ and runs in deterministic time $O((p +n)^{1+o(1)} + \sum_{i=1}^n |A_i|)$. 
\end{lemma}

\begin{proof}
If the termination criterion $\sum_{j \leq i } |X_{r,j} | > p + i -1$ is satisfied, then Lemma~\ref{lem:cauchy} implies that $X_{r,1}+\ldots+X_{r,i} = \Z_p$, and hence $A_1+\ldots+A_n = \Z_p$, so we correctly return $\Z_p$. 
Otherwise we reach the last line of Algorithm~\ref{alg:prime_sum}, and we correctly computed $X_{\log n, 1} = A_1 + \ldots + A_n$. This shows correctness.

To analyze the running time, let $(r^\ast,i^\ast)$ be the values of $r$ and $i$ at the end of the execution of the algorithm. In particular, if $r^\ast =  \log n $ we have $i^\ast = 1$. By our use of Theorem \ref{thm:sumset_deterministic}, the total running time of the algorithm is 
\begin{align*}
\sum_{r <r^\ast} \sum_{i=1}^{n/2^r  } |X_{r,i}| \cdot p^{o(1)} + \sum_{i \leq i^\ast} |X_{r^\ast,i}| \cdot p^{o(1)}.
\end{align*}
We use the fact that the termination criterion was not satisfied before step $(r^\ast,i^\ast)$ to obtain:
\begin{align*}
  \sum_{i =1}^{n/2^r} |X_{r,i}| &\leq p+ \frac n{2^r} -1 \qquad \text{for any }r < r^\ast,  \\
  \sum_{i < i^\ast} |X_{r^\ast,i}| &\leq p + \frac{n}{2^{r^\ast}} -1.
\end{align*}
Moreover, we have $|X_{r^\ast,i^\ast}| \le p$. Combining these observations allows us to further bound the running time by
\[ \left( \sum_{r <r^\ast} \left(p+ \frac n{2^r} -1\right) + \left(p + \frac{n}{2^{r^\ast}} -1 \right) + p \right) \cdot p^{o(1)} = (p \log n + n) \cdot p^{o(1)} = (p+n)^{1+o(1)}. \qedhere \]
\end{proof}

%


\section{Algorithm for \boldmath$n$-Fold Sumset Computation}
\label{sec:n_fold_general}


This section proves Theorem~\ref{thm:n_fold_comp}. The main idea is that whenever we detect a non-trivial symmetry group we reduce to a problem over a smaller universe $\Z_d$, for a divisor $d$ of $m$. 

Consider Algorithm~\ref{alg:general_sum}. Suppose we are given sets $A_1,\ldots,A_n \subseteq \Z_p$. We may assume that $n$ is a power of $2$, since otherwise we can add an appropriate number of sets $A_i = \{0\}$ without affecting the sumset. 
We maintain a guess $s$ of the outputsize $|A_1+\ldots+A_n|$. Specifically, $s$ loops over all powers of $2$ starting from $2^0 = 1$, and the algorithm returns the correct result once we reach the first iteration with $s \ge |A_1+\ldots+A_n|$.
Thus, in iteration~$s$ we know that the output size is more than $s/2$, and our primary goal is to test whether $|A_1+\ldots+A_n| \le s$. If this is true then we want to compute the set $A_1+\ldots+A_n$.
We compute $A_1+\ldots+A_n$ in a tree-like bottom-up fashion, by first computing $A_1+A_2, A_3+A_4,\ldots$, then computing $A_1+A_2+A_3+A_4, \ldots$, and so on. The intermediate sets in round $r$ are called $X_{r,1},\ldots,X_{r,n/2^r}$. Our two main ideas now are as follows.

First, due to the presence of non-trivial subgroups in $\Z_m$ when $m$ is not a prime, an intermediate set~$X_{r,i}$ can have a non-trivial symmetry group $Sym(X_{r,i})$. 
As a criterion for a non-trivial symmetry group we test whether $|X_{r,i}| < |X_{r-1,2i}| + |X_{r-1,2i+1}| -1$ (cf.\ Corollary~\ref{cor:kneser} of Kneser's Theorem). 
Once we have found a non-trivial symmetry group $Sym(X_{r,i})$, then also $Sym(A_1+\ldots+A_n) \supseteq Sym(X_{r,i})$ is non-trivial, and thus the output set $A_1+\ldots+A_n$ is periodic, with period length $d = m / |Sym(X_{r,i})|$. It therefore suffices to compute $A_1+\ldots+A_n$ modulo $d$. Hence, we reduce to a problem over a smaller universe $\Z_d$. This case is handled in lines~\ref{line:find_symmetry_group}-\ref{line:recursion}. Note that $d$ may not be the smallest period length for $A_1+A_2+\ldots+A_n$, but since we only need to reduce the problem size, any period suffices for us.

Second, if the criterion $|X_{r,i}| < |X_{r-1,2i}| + |X_{r-1,2i+1}| -1$ is never satisfied, then we can use it to bound the output size. Specifically, we obtain a lower bound for $|X_{\log n,1}|$ in terms of the total intermediate size $\sum_j |X_{r,j}|$. In particular, if the total intermediate size is much larger than $s$, then also the output size is more than $s$. However, we cannot move to the next guess $2s$ yet, since we do not know whether the criterion $|X_{r,i}| < |X_{r-1,2i}| + |X_{r-1,2i+1}| -1$ will be satisfied in future rounds $r' > r$. Nevertheless, we argue that once we have intermediate set size $\sum_{j\le i} |X_{r,j}| \gg s$, then we can ignore the remaining sets $X_{r,j}$, $j > i$, by setting them to $\{0\}$, cf.\ lines~\ref{line:sum_condition}-\ref{line:last_extra_line}. This allows us to bound the total size of all intermediate sets to be linear in the input plus output size.

\begin{algorithm}[!t]\caption{}\label{alg:general_sum}
\begin{algorithmic}[1]
\Procedure{\textsc{nFoldSumset}}{$A_1,\ldots,A_n,m$} \State \Comment{$n$ is a power of 2; non-empty $A_1,\ldots,A_n \subseteq \Z_m$}
\For {$s =1, 2, 4, \ldots, 2^{\left \lceil \log m \right \rceil  }$ }
	\State $X_{0,i} \leftarrow A_i$, for all $i \in [n]$ \label{line:initialize_x}
	\For { $r=1$ to $\log n$}
		\For { $i=1$ to $n/2^r$}
			\State $X_{r,i} \leftarrow X_{r-1,2i-1} + X_{r-1,2i}$ \label{line:update_x} \Comment{sumset computation via Theorem \ref{thm:sumset_deterministic} or \ref{thm:sumset_randomized}}
		\If { $|X_{r,i}| < |X_{r-1,2i-1}| + |X_{r-1,2i}| - 1$ } \label{line:find_symmetry_group}
			\State Compute $Sym(X_{r,i}) $ \Comment{symmetry group computation via Theorem \ref{thm:symmetry_group}}
			\State $d \leftarrow m / |Sym(X_{r,i})|$ \Comment{ $Sym(X_{r,i}) = d \cdot \mathbb{Z}_{m/d}$}
			\State $A_i' \leftarrow A_i \bmod d$, for all $i \in [n]$
			\State \Return $ \textsc{nFoldSumset}(A_1', \ldots, A_n', d) + d\cdot \{0,1,2,\ldots,m/d-1\}$ \label{line:recursion}
		\EndIf
		\If { $\sum_{j \leq i} |X_{r,j}| \ge s+n/2^r $ } \label{line:sum_condition}
			\State $X_{r,j} \leftarrow \{0\}$, for all $i < j \le n/2^r$ \label{line:set_to_zero_x}
			\State \textbf{break} \label{line:last_extra_line}
		\EndIf
		\EndFor
	\EndFor
	\If{ $|X_{\log n, 1 }| \le s$} \label{line:leaf_condition}
		\State \Return $X_{\log n,1}$ \label{line:return}
	\EndIf 
\EndFor
\EndProcedure
\end{algorithmic}
\end{algorithm}

We next prove correctness and then analyze the running time of Algorithm~\ref{alg:general_sum}.

\begin{lemma}[Correctness of Algorithm~\ref{alg:general_sum}] \label{lem:correctness}
  Given non-empty sets $A_1,\ldots,A_n \subseteq \Z_m$, where $n$ is a power of 2, Algorithm~\ref{alg:general_sum} correctly computes $A_1+\ldots+A_n$. 
\end{lemma}

\begin{proof} 
Note that without lines~\ref{line:sum_condition}-\ref{line:last_extra_line}, 
we would compute the sumset in a straightforward bottom-up tree-like fashion as $((A_1+A_2)+(A_3+A_4))+\ldots$, and thus the intermediate set $X_{r,i}$ would be equal to $A_x + A_{x+1} + \ldots + A_y$ for $x = (i-1) 2^r + 1$ and $y = i 2^r$. 
In the additional lines~\ref{line:sum_condition}-\ref{line:last_extra_line}, we set some intermediate sets $X_{r,i}$ to $\{0\}$. Thus, we may lose some summands, but any intermediate set~$X_{r,i}$ still corresponds to the sumset of a subset of its summands $A_x,A_{x+1},\ldots,A_y$. More precisely, the set $X_{r,i}$ satisfies $X_{r,i} = A_{z_1} + A_{z_2} + \ldots + A_{z_\ell}$ for some $\{z_1,\ldots,z_\ell\} \subseteq \{ x, x+1, \ldots, y\}$, with the understanding that $X_{r,i} = \{0\}$ if $\ell = 0$.
(This property holds initially in line~\ref{line:initialize_x} and it continues to hold when we set $X_{r,i}$ in lines~\ref{line:update_x} and \ref{line:set_to_zero_x}.)
In particular, we always have 
\begin{align} \label{eq:XSymSubset} 
  Sym(X_{r,i}) = Sym(A_{z_1} + A_{z_2} + \ldots + A_{z_\ell}) \subseteq Sym(A_1+\ldots+A_n).
\end{align}
Moreover, we also infer
\begin{align} \label{eq:XsizeUpperBound} 
  |X_{r,i}| = |A_{z_1} + A_{z_2} + \ldots + A_{z_\ell}| \le |A_1+\ldots+A_n|. 
\end{align}

We shall perform induction on the universe size $m$. For the base case $m=1$, the result is obvious. For larger $m$, we consider the following two cases.

\subparagraph{Case 1:} \emph{At some point in the execution, the criterion $|X_{r,i}| < |X_{r-1,2i-1}| + |X_{r-1,2i}| - 1$ in line~\ref{line:find_symmetry_group} is satisfied.}
Then by Corollary~\ref{cor:kneser}, $Sym(X_{r,i})$ is non-trivial and hence $Sym( A_1+\ldots+A_n) \supseteq Sym(X_{r,i})$ is also non-trivial. We make use of the fact that all subgroups of $\Z_m$ are of the form $d \cdot \Z_{m/d}$, where $d$ divides $m$. In particular, $Sym(X_{r,i}) = d \cdot \Z_{m/d}$ for $d := m / |Sym(X_{r,i})|$. This means that $A_1+\ldots+A_n$ is cyclic with period length $d$. It follows that for $A_i' := A_i \bmod d$ we have (using the induction hypothesis on $d$)
\[ A_1+\ldots+A_n = \textsc{nFoldSumset}(A_1',\ldots,A_n',d) + d \cdot \big\{0,1,\ldots,\tfrac md - 1\big\}. \]
This shows correctness of lines \ref{line:find_symmetry_group}-\ref{line:recursion}.

\subparagraph{Case 2:} \emph{Lines \ref{line:find_symmetry_group}-\ref{line:recursion} are never executed.} That is, for each computed set $X_{r,i}$ in line~\ref{line:update_x} we have
\begin{align} \label{eq:Xsizebound}
  |X_{r,i}| \ge |X_{r-1,2i-1}| + |X_{r-1,2i}| - 1.
\end{align}
We use inequality (\ref{eq:Xsizebound}) to analyze line~\ref{line:sum_condition}. 
Fix any $s \in \{1,2,4,\ldots,2^{\lceil \log m \rceil}\}$, and consider iteration~$s$. 

\begin{claim} \label{cla:iffX}
In iteration $s$, we have $|X_{\log n,1}| > s$ if and only if $|A_1+\ldots+A_n| > s$. Moreover, if $|X_{\log n,1}| \le s$ then $X_{\log n,1} = A_1+\ldots+A_n$.
\end{claim}

Comparing this claim with lines~\ref{line:leaf_condition}-\ref{line:return}, we see that if our guess $s$ for the output size is too small, i.e., $|A_1+\ldots+A_n| > s$, then the algorithm proceeds with the next larger guess. Otherwise, the algorithm correctly computes $X_{\log n,1} = A_1+\ldots+A_n$ and returns this set. It remains to prove the claim.

\begin{proof}
The `only if' part follows from the bound $|X_{\log n,1}| \le |A_1+\ldots+A_n|$ by (\ref{eq:XsizeUpperBound}).

For the `if' part, we consider two cases:

\subparagraph{Case A:} If the criterion $\sum_{j \leq i} |X_{r,j}| \ge s+n/2^r$ in line~\ref{line:sum_condition} is never satisfied, then the algorithm computes $X_{\log n,1} = A_1+\ldots+A_n$ in a straightforward manner, and thus $|X_{\log n,1}| = |A_1+\ldots+A_n|$. 

\subparagraph{Case B:} If the criterion $\sum_{j \leq i} |X_{r,j}| \ge s+n/2^r$ in line~\ref{line:sum_condition} is satisfied in some iteration $r$, then the following bound shows that it will also be satisfied in iteration $r+1$ for some value of $i$:
\[ \sum_{j=1}^{n/2^{r+1}} |X_{r+1,j}| \stackrel{(\ref{eq:Xsizebound})}{\ge} \sum_{j=1}^{n/2^r} |X_{r,j}| - \frac n{2^{r+1}} \ge \Big(s + \frac n{2^r}\Big) - \frac n{2^{r+1}} = s + \frac n{2^{r+1}}. \]
This nearly proves that the criterion is satisfied in iteration $r+1$, but it ignores that some of the sets $X_{r+1,j}$ could be set to $\{0\}$ by lines~\ref{line:sum_condition}-\ref{line:last_extra_line}. However, when this happens then by the criterion in line~\ref{line:sum_condition} we nevertheless have 
  $\sum_{j=1}^{n/2^{r+1}} |X_{r+1,j}| \ge s+n/2^{r+1}$.

Therefore, if the criterion in line~\ref{line:sum_condition} is satisfied in some iteration $r$, then it is also satisfied for $r = \log n$, which yields $|X_{\log n,1}| \ge s+1 > s$.
 
\medskip
In either case, we obtain $|X_{\log n,1}| > s$ if $|A_1+\ldots+A_n| > s$. This proves the equivalence. 

For the second claim, note that $|X_{\log n,1}| \le s$ only happens in Case A, and in this case we showed that $X_{\log n,1} = A_1+\ldots+A_n$.
\end{proof}

\medskip
In summary, if at any point during the course of the algorithm the criterion $|X_{r,i}| < |X_{r-1,2i-1}| + |X_{r-1,2i}| - 1$ in line~\ref{line:find_symmetry_group} is satisfied (Case 1), then we have found a non-trivial symmetry group, and we can move to a problem over a smaller universe $\Z_d$, where $d < m$ is a divisor of $m$. Correctness then follows by induction on $m$. If this never happens (Case 2), then the algorithm behaves as follows. We have an increasing guess $s$ for the output size $|A_1+\ldots+A_n|$. When this guess is too small, at some point the criterion $\sum_{j \leq i} |X_{r,j}| \ge s+n/2^r $ in line~\ref{line:sum_condition} is satisfied, from which point on we set some of the sets $X_{r,i}$ to $\{0\}$, but we ensure that we end up with $|X_{\log n,1}| > s$. This allows us to conclude that our guess $s$ was too small, so we increase it. When our guess $s$ reaches the smallest power of $2$ that is at least $|A_1+\ldots+A_n|$, then the algorithm correctly computes $X_{\log n,1} = A_1+\ldots+A_n$ and returns this set.
\end{proof}

\begin{lemma}[Running Time of Algorithm~\ref{alg:general_sum}] \label{lem:runningtimegeneral}
  Let $k := |A_1|+\ldots+|A_n| + |A_1+\ldots+A_n|$ be the total input plus output size. Depending on whether we use Theorem \ref{thm:sumset_deterministic} or Theorem~\ref{thm:sumset_randomized} for sumset computation, Algorithm~\ref{alg:general_sum} is
  \begin{enumerate}
    \item deterministic and runs in time $k \cdot 2^{ O(\sqrt{ \log k \log \log m}) } \cdot \log m$, or
    \item randomized and runs in expected time $O(k \cdot \polylog (mk))$.
  \end{enumerate}
\end{lemma}
\begin{proof} Let $T(k,m)$ be the running time of our algorithm.
Note that we have at most one call to a recursive subproblem in line~\ref{line:recursion}, incurring time $T(k,d)$, where $d$ is a divisor of $m$ and thus $d \le m/2$. 

Let $s^\ast$ be the smallest power of $2$ that is at least $|A_1+\ldots+A_n|$.
Similarly as in the proof of correctness, we see that the algorithm only performs iterations $s$ from $1$ to at most $s^\ast$, since we return the correct output in iteration $s^\ast$, unless we call a recursive subproblem before that.

We bound the running time in iteration $s$ as follows. For any iteration $r$, let $X_{r,i(r)}$ be the last set that we computed in line~\ref{line:update_x}. (That is, after computing $X_{r,i(r)}$ we either move to a recursive call, or we set all remaining sets $X_{r,j} \leftarrow \{0\}$, for any $j > i(r)$.)
Note that $\sum_{j < i(r)} |X_{r,j}| < s + n/2^r$, since otherwise we would have set $X_{r,i(r)} \leftarrow \{0\}$ and not computed it in line~\ref{line:update_x}. Moreover, $|X_{r,i(r)}| \le |A_1+\ldots+A_n| \le k$ by~(\ref{eq:XsizeUpperBound}). By Theorem~\ref{thm:sumset_deterministic}, computing $X_{r,i}$ takes time 
\[|X_{r,i}| \cdot 2^{ O(\sqrt{ \log |X_{r,i}| \log \log m}) } \le |X_{r,i}| \cdot 2^{ O(\sqrt{ \log k \log \log m}) } . \] 
Therefore, the total time spent in iteration $r$ is bounded by 
\[ \left(s + \frac n{2^r} + k\right) \cdot 2^{ O(\sqrt{ \log k \log \log m}) } \le k \cdot 2^{ O(\sqrt{ \log k \log \log m}) }, \]
for any $1 \le s \le s^\ast = O(k)$. Summing over all iterations $r$ adds a factor $\log n \le \log k \le 2^{O(\sqrt{\log k})}$, which can be ignored. Summing over all iterations $s$ adds a factor $\log s^\ast = O(\log k)$, which can also be ignored. Adding the potential recursive call, the total running time is 
\[ T(k,m) \le k \cdot 2^{ O(\sqrt{ \log k \log \log m}) }  + T(k,m/2). \]
This solves to total time $k \cdot 2^{ O(\sqrt{ \log k \log \log m}) } \cdot \log (m)$. 

The analysis of the randomized variant is analogous.
\end{proof}


\begin{proof}[Proof of Theorem~\ref{thm:n_fold_comp}]
Algorithm~\ref{alg:general_sum} almost proves the theorem, except that the deterministic running time shown in Lemma~\ref{lem:runningtimegeneral} is $k \cdot 2^{ O(\sqrt{ \log k \log \log m}) } \cdot \log m$ instead of the promised $k \cdot 2^{ O(\sqrt{ \log k \log \log m}) }$. Note that the former can be bounded by the latter unless $k \le \log^c m$, for some absolute constant $c$. 
    In the case $k \le \log^c m$ we switch to a different algorithm. Specifically, we simply compute $((A_1+A_2)+A_3)+\ldots +A_n$ in a linear fashion, in each step using a naive sumset computation that computes $A+B$ in time $\widetilde O(|A| \cdot |B|)$. Since each intermediate result has size at most $k$, each sumset computation takes time $O(k^2)$. Since $n \le k$, in total this simple algorithm runs in time $O(k^3)$. Finally, since $k \le \log^c m$ we can bound $O(k^3) \le 2^{O(\sqrt{\log k \log \log m})}$. This shows the promised running time also in case $k \le \log^c m$. We obtain the promised guarantees even if we do not know $k$, by running both algorithms in parallel until the first one finishes.
\end{proof}

\section{Output-sensitive Sumset Computation}
\label{sec:sumset_comp}

Recall that in sumset computation we are given sets $A,B \subseteq \Z_m$ and the task is to compute $A+B$. 
In this section we present a deterministic algorithm for sumset computation. We also show a generalization to convolution of non-negative vectors, proving Theorem~\ref{thm:sumset_deterministic}.

Chan and Lewenstein~\cite{CL15} designed very efficient algorithms for sumset computation in a specialized setting, in which the input additionally contains a set $T$ promised to be a superset of $A+B$. Their running time is close to linear in $|T|$. Specifically, they proved the following lemma.

\begin{lemma}[FFT Lemma from~\cite{CL15}] \label{lemma:chan_lewenstein}
Given sets $A,B \subseteq \{0,1,\ldots,m-1\}$ and given a set $T$ which is known to be a superset of $A+B$, we can compute $A+B$ (over $\Z$)
\begin{enumerate}
\item[(1)] by a randomized Las Vegas algorithm in $O(|T| \polylog m)$ expected time, or
\item[(2)] by a deterministic algorithm in $|T| \cdot 2^{O(\sqrt{\log |T| \log \log m})}$ time
\end{enumerate}
The running time bounds are taken from \cite[Section 8]{CL15}.
\end{lemma}

Here we show a trick that yields the same time bounds in the standard setting (without the additional set $T$). We note that it makes no significant difference whether we compute $A+B$ over $\Z_m$ or over $\Z$, as discussed also in the introduction. We choose to work over $\Z_m$, for consistency with the rest of this paper.

\begin{lemma} \label{lemma:sumset_both}
Given sets $A,B \subseteq \Z_m$, we can compute $A+B$ (over $\Z_m$)
\begin{enumerate}
\item[(1)] by a randomized Las Vegas algorithm in $O(|A+B| \polylog m)$ expected time, or
\item[(2)] by a deterministic algorithm in $|A+B| \cdot 2^{O(\sqrt{\log |A+B| \log \log m})}$ time.
\end{enumerate}
\end{lemma}

Note that bullet point (1) reproves Theorem~\ref{thm:sumset_randomized} by Cole and Hariharan~\cite{CH02}, and bullet point~(2) answers an open problem by Chan and Lewenstein~\cite{CL15}.

\begin{proof}
  First note that we can assume $m$ to be a power of 2. Indeed, if $m$ is not a power of 2, we let $m'$ be the smallest power of 2 greater than $2m$. Given $A,B \subseteq \{0,1,\ldots,m-1\}$ we compute $A+B$ over $\Z_{m'}$ and take the resulting set modulo $m$ to obtain $A+B$ over $\Z_m$. This assumption is not necessary, but shall make the exposition cleaner, avoiding using the ceil and floor functions.
  
  So assume that $m$ is a power of 2, and set $m' := m/2$.
  Let $A' := A \bmod m'$ and $B' := B \bmod m'$ and recursively compute $S := A' + B'$ over $\Z_{m'}$. 
  Then we have $S = (A+B) \bmod m'$. Thus, since $\max(A)+\max(B) < 2m \le 4m'$, the set $T := S + \{0,m',2m',3m'\}$ covers $A+B$. In other words, $T$~is a superset of $A+B$ over $\Z$. We can thus use the FFT Lemma to compute $A+B$ over $\Z$. Reducing the resulting set modulo $m$ yields $A+B$ over $\Z_m$. This leads to the recursive Algorithm~\ref{alg:deterministic_sumset}. 
  
  Since we can bound $|T| \le 4|S| \le 4|A+B|$, the expected running time of one recursive step is $O(|A+B| \cdot \polylog m)$, and there are $O(\log m)$ recursive steps. This yields the claimed randomized running time. For the deterministic variant we obtain running time $|A+B| \cdot 2^{O(\sqrt{\log |A+B| \log \log m})} \cdot \polylog m$ from the FFT Lemma, times an additional $\log m$ factor due to the recursion. 

Now, to get rid of the additional $\polylog(m)$ factor and obtain the promised guarantee, we shall observe the following. If $|A+B| \ge \log m$, then this running time is bounded by the claimed $|A+B| \cdot 2^{O(\sqrt{\log |A+B| \log \log m})}$. If $|A+B| < \log m$, then the naive approach which computes $A+B$ in time $\tilde O(|A| \cdot |B|) = \tilde O(|A+B|^2) = 2^{O(\sqrt{\log |A+B| \log \log m})}$. Running both algorithms in parallel until the first one finishes yields the claimed bound.
\end{proof}

\begin{algorithm}[!t]\caption{}\label{alg:deterministic_sumset}
\begin{algorithmic}[1]
\Procedure{\textsc{{DeterministicSumset}}}{$A,B,m$} \\
\Comment{$m$ is a power of 2; non-empty $A,B \subseteq \Z_m$; computes $A+B$ over $\Z_m$}
\State $m' := m/2$
\State $S \leftarrow \textsc{DeterministicSumset}(A \bmod m', B \bmod m', m')$
\State $T \leftarrow S + \{0,m',2m',3m'\}$ over $\Z$ \Comment{$T \supseteq A+B$ over $\Z$}
\State Compute $R := A+B$ over $\Z$ via the FFT Lemma using additional input $T$
\State \Return $R \bmod m$
\EndProcedure
\end{algorithmic}
\end{algorithm}

A similar result also holds for convolution of non-negative vectors. We denote by $\|x\|_0$ the number of non-zero entries of a vector $x$.

\begin{lemma} \label{lemma:nonneg_conv}
Given vectors $A,B \in \mathbb{R}_{\ge 0}^m$, we can compute their convolution $A \conv B$ (with wraparound)
\begin{enumerate}
\item[(1)] by a randomized Las Vegas algorithm in $O(\|A \conv B\|_0 \polylog m)$ expected time, or
\item[(2)] by a deterministic algorithm in $\|A \conv B\|_0 \cdot 2^{O(\sqrt{\log \|A \conv B\|_0 \log \log m})}$ time.
\end{enumerate}
\end{lemma}
Again bullet point (1) reproves a result by Cole and Hariharan~\cite{CH02}, and bullet point~(2) proves Theorem~\ref{thm:sumset_deterministic}.
\begin{proof}
  Denote by $I$ and $J$ the indicator vectors of the non-zero entries of $A$ and $B$, respectively. Observe that $|I+J| = \|A\star B\|_0$. We can thus compute $I+J$ in expected time $O(\|A\star B\|_0 \polylog m)$ by Lemma~\ref{lemma:sumset_both}. We now make use of a variant of the FFT Lemma from~\cite[Remark 8.2]{CL15}, stating that if we know a superset $T \supseteq I+J$ then we can compute $A \conv B$ in expected time $O(|T| \polylog m)$. Using this for $T=I+J$ yields expected time $O(\|A\star B\|_0 \polylog m)$, or time $\|A\star B\|_0 \cdot 2^{O(\sqrt{\log \|A\star B\|_0 \log \log m})} \cdot \poly(\log m)$ for the deterministic variant. Now, we can get rid of the $\polylog(m)$ factors in the deterministic variant using the same argument as the one in Lemma~\ref{lemma:sumset_both}.
\end{proof}

\section{Computing the Symmetry Group}
\label{sec:symmetry_group}

In this section, we show how to compute the symmetry group $Sym(A)$ for any given non-empty set $A \subseteq \Z_m$ in time $O(|A|)$, proving Theorem \ref{thm:symmetry_group}.
Let $n := |A|$ and denote by $a_1 < a_2 < \ldots < a_n$ the elements of~$A$. For simplicity of notation, we set 
\[ a_{n+1} := a_1, \;\; a_{n+2} := a_2, \;\; \ldots, \;\; a_{2n} := a_n. \]
Note that for our applications of this Theorem, $a_i$ correspond to residue classes modulo $m$.
We construct a string $P$ (the pattern) of length $n$ by setting for any $1 \le i \le n$:
\[ P_i := (a_{i+1} - a_i) \bmod m. \]

Similarly, we construct a string $T$ (the text) of length $2n-1$ by setting for any $1 \le i \le 2n-1$:
\[ T_i := (a_{i+1} - a_i) \bmod m. \]
Note that the text is constructed by repeating the pattern twice and removing the last letter. 
%
%

%

We say that there is a match of pattern $P$ in text $T$ at position $i$ if $P_j = T_{i-1+j}$ holds for any $1 \le j \le n$. The following lemma shows that the matches of $P$ in $T$ are in one-to-one correspondence with the symmetry group $Sym(A)$. Since all matches of $P$ in $T$ can be computed in time $O(n)$ by the classic Knuth-Morris-Pratt pattern matching algorithm, this finishes the proof of Theorem~\ref{thm:symmetry_group}.
%

\begin{lemma}
If there is a match of $P$ in $T$ at position $i$, then $a_i - a_1 \in Sym(A)$. Moreover, for any $x \in Sym(A)$, we have $x = a_i-a_1$ for some $1 \le i \le n$ and there is a match of $P$ in $T$ at position $i$.
\end{lemma}

\begin{proof}
Note that there is a match at position $i$ if and only if for all $1 \leq j \leq n$
\[
a_{j+1} - a_j = a_{i+j} - a_{i+j-1} \pmod m.
\]
Summing this equation in a telescoping sum over all $j \in \{1,\ldots,\ell-1\}$, for any fixed $1 \le \ell \le n$, yields
\[ a_{\ell} - a_1 = a_{i+\ell-1} - a_i \pmod m, \]
or, equivalently,
\[ a_{\ell} + (a_i - a_1) = a_{i+\ell-1} \pmod m. \]
This establishes $a_i - a_1 \in Sym(A)$. 

For the second part, recall that $Sym(A) \subseteq A - \{a_1\}$, as discussed in Section~\ref{sec:toolkit_sym}. Therefore for any $x \in Sym(A)$ we have $x = a_i  - a_1$ for some $i$. 
Now consider the values $a'_j := (a_j - x) \bmod m$ for $1 \le j \le n$. Observe that the sequence $a'_1,\ldots,a'_n$ is monotonically increasing up to some point, where the modulo operation reduces by an additional $-m$, and then is again monotonically increasing. In particular, for some $1 \le r \le n$ we have
\[ a'_r < a'_{r+1} < \ldots < a'_{n-1} < a'_n < a'_1 < a'_2 < \ldots < a'_{r-2} < a'_{r-1}. \]

Since $x \in Sym(A)$ and $Sym(A)$ is a group, also $-x \in Sym(A)$, and thus $a'_j \in A$ for all~$j$. Hence, the $n$ different values $a'_j$ must correspond to the elements of $A$. 
It follows that $a'_{r-1+j} = a_j$ for all $1 \le j \le n$.

Observing that $a'_i = a_i - x = a_i - (a_i - a_1) = a_1 \pmod m$, we see that $r = i$. 
In other words, we have for all $1 \le j \le n$
\[ a_{i-1+j} - (a_i - a_1) = a_j \pmod m. \]
Subtracting this equation for $j$ from this equation for $j+1$ yields, for any $1 \le j \le n$,
\[ a_{i+j} - a_{i+j-1} = a_{j+1} - a_j \pmod m. \]
As noted in the beginning of this proof, this means that there is a match of $P$ in $T$ at position~$i$.
\end{proof}


\bibliographystyle{alpha}
\bibliography{ref}


\end{document}